\documentclass[a4paper,12pt]{amsart}
\usepackage{amsaddr}
\usepackage[left=3.4 cm,top=2.4cm,bottom=3.4cm,right=3.4cm]{geometry}

\usepackage{xspace}

\usepackage{amsfonts,amssymb,amsthm,amsmath}
\usepackage{bbm}

\usepackage{xcolor}

\definecolor{green(munsell)}{rgb}{0.0, 0.66, 0.47}
\definecolor{BlueGreenn}{rgb}{0.3,0.5,0.8}
\definecolor{DB}{rgb}{0.3,0.3,0.3}
\definecolor{DOr}{rgb}{0.7,0.3,0.3}
\definecolor{DGr}{rgb}{0.3,0.7,0.3}
\definecolor{DBl}{rgb}{0.1,0.3,0.5}
\definecolor{arylideyellow}{rgb}{0.91, 0.84, 0.42}
\definecolor{burntorange}{rgb}{0.8, 0.33, 0.0}
\definecolor{chromeyellow}{rgb}{1.0, 0.65, 0.0}
\usepackage[hypertexnames=true, citecolor = burntorange, colorlinks = true, linkcolor = BlueGreenn, pdffitwindow = false, urlbordercolor = white,pdfborder={0 0 0}]{hyperref}%

\usepackage[T1]{fontenc}
\usepackage{microtype}

\usepackage{color}
\usepackage[utf8]{inputenc}

\usepackage{framed}
\usepackage{centernot}

\usepackage{tikz-cd} 
\usepackage{tkz-euclide}

\numberwithin{equation}{section}

\newtheorem{theorem}{Theorem}[section]
\newtheorem{proposition}[theorem]{Proposition}
\newtheorem{lemma}[theorem]{Lemma}

\newtheorem{definition}[theorem]{Definition}

\newcommand{\ot}{\otimes}

\DeclareMathOperator{\ad}{ad}

\newcommand{\Cl}{\mathbb{C}}
\newcommand{\Rl}{\mathbb{R}}
\newcommand{\Nl}{\mathbb{N}}

\newcommand{\B}{\mathcal{B}}

\newcommand{\M}{\mathcal{M}}
\newcommand{\N}{\mathcal{N}}
\newcommand{\Hil}{\mathcal{H}}

\newcommand{\Om}{\Omega}
\newcommand{\om}{\omega}
\newcommand{\la}{\lambda}
\newcommand{\eps}{\varepsilon}
\newcommand{\te}{\theta}

\usepackage{mathtools, mathrsfs}

\usepackage{ytableau} 
\ytableausetup{boxsize=0.5em,centertableaux}



\usepackage{tikz}

\newcommand{\Std}{{\rm (A)}\xspace}
\newcommand{\Int}{{\rm (B)}\xspace}
\newcommand{\Sing}{{\rm (C)}\xspace}
\newcommand{\Hardy}{\mathbb H}
\newcommand{\Strip}{\mathbb S}
\newcommand{\PG}{\mathcal{P}}
\newcommand{\bte}{\boldsymbol{\theta}}

\title[Deformations of half-sided modular inclusions]{Deformations of half-sided modular inclusions and non-local chiral field theories}

\author{Gandalf Lechner}
\address{Department Mathematik, FAU Erlangen-Nürnberg, Germany\\ gandalf.lechner@fau.de}
\author{Charley Scotford}
\address{School of Mathematics, Cardiff University, United Kingdom\\ scotfordc@cardiff.ac.uk}
\date{November 04, 2021}

\newcommand{\scrX}{{\mathscr{A}_\infty}}
\newcommand{\POm}{P_{\Om}}
\newcommand{\Ploc}{P_{\rm loc}}

\newcommand{\Hilloc}{{\mathcal H}_{\rm loc}}
\newcommand{\Pst}{{\mathrm{P}}}

\newcommand{\Ss}{{\mathscr S}}
\newcommand{\La}{\Lambda}

\def\ker{\operatorname{ker}}

\def\dim{\operatorname{dim}}

\DeclareMathOperator*{\wlim}{w-lim}
\DeclareMathOperator*{\slim}{s-lim}

\def\Im{\operatorname{Im}}

\newcommand{\I}{\mathcal{I}}
\newcommand{\D}{\mathcal{D}}

\newcommand{\CA}[0]{\mathcal{A}} \newcommand{\CB}[0]{\mathcal{B}}

\newcommand{\CI}[0]{\mathcal{I}}

\newcommand{\CO}[0]{\mathcal{O}}

\begin{document}
	 
\maketitle

\begin{abstract}
    We construct explicit examples of half-sided modular inclusions $\N\subset\M$ of von Neumann algebras with trivial relative commutants. After stating a general criterion for triviality of the relative commutant in terms of an algebra localized at infinity, we consider a second quantization inclusion $\N\subset\M$ with large relative commutant and construct a one-parameter family $\N_\kappa\subset\M_\kappa$, $\kappa\geq0$, of half-sided inclusions such that $\N_0=\N$, $\M_0=\M$ and $\N_\kappa'\cap\M_\kappa=\Cl1$ for $\kappa>0$. The technique we use is an explicit deformation procedure (warped convolution), and we explain the relation of this result to the construction of chiral conformal quantum field theories on the real line and on the circle.
\end{abstract}

\section{Introduction}

In the operator-algebraic approach to quantum field theory \cite{Haag:1996,Araki:1999}, models of quantum field theories on a spacetime $\mathscr M$ are described by assigning to open regions $\CO\subset\mathscr M$ von Neumann algebras $\CA(\CO)$ that act on a common Hilbert space and are subject to various interrelated inclusion, commutation, covariance, and spectral properties. While this setup implements the physical principles underlying quantum field theory clearly and rigorously, it is usually difficult to find examples of models that satisfy the axioms and are non-trivial. This is especially the case if ``non-trivial'' is understood to mean that the considered quantum field theory should describe particles with non-trivial interaction. 

In view of this problem, several researchers have considered indirect descriptions of quantum field theories by simpler data. A possible point of view is to fix one of the local algebras $\M=\CA(\CO_0)$, generally isomorphic \cite{BuchholzFredenhagenDAntoni:1987} to the unique hyperfinite factor of type III${}_1$ \cite{Haagerup:1987} and then aim at constructing the net $\CO\mapsto\CA(\CO)$ of all local algebras with the help of group actions, generating von Neumann algebras, and relative commutants -- see \cite{BaumgrtelWollenberg:1992,GuidoLongoWiesbrock:1998,BuchholzLechner:2004,BuchholzLechnerSummers:2011,Tanimoto:2013,BarataJakelMund:2013} for various implementations of this and related ideas.

In these approaches, it is often of central importance to make sure that the relative commutant of an inclusion $\CA(\CO)\subset\CA(\tilde\CO)$ is large in some sense or at least non-trivial. However, such an analysis of relative commutants of von Neumann algebras is often a delicate problem as there are typically no means to directly construct any of its elements. Moreover, the inclusions typically encountered in quantum field theory are not of the form typically studied in operator algebras: The algebras involved are type III${}_1$ factors, and usually no meaningful notion of index or conditional expectation exists. It is therefore important to find manageable examples in which the structure of the relative commutant can be determined.

In this article we investigate quantum field theories on one of the simplest spacetimes, namely the real line ${\mathscr M}=\Rl$, to be thought of as a chiral half of a conformal field theory on a light ray $\Rl\subset\Rl^2$. Such QFTs are closely related to so-called half-sided modular inclusions of von Neumann algebras (Def.~\ref{def:HSM}) which, under favourable circumstances, encode the full structure of the field theory. As we will recall in Section~\ref{section:3cases}, half-sided modular inclusions $\N\subset\M$ (or, equivalently, one-dimensional Borchers triples) can be grouped into three broad families \Std, \Int, \Sing according to the size of their relative commutant $\N'\cap\M$. Here case \Sing corresponds to triviality of the relative commutant $\N'\cap\M=\Cl$. Whereas examples of case \Std and \Int are long since known, the first example of a singular (case \Sing) inclusion was found only very recently by Longo, Tanimoto, and Ueda with methods from free probability \cite{LongoTanimotoUeda:2019}.

The main purpose of this article is to give new examples of singular half-sided inclusions that can be analyzed in detail. To do so, we proceed as follows: In Section~\ref{section:algebrainfinity}, we give a general criterion for determining case \Sing which is based on the algebra localized at the point at infinity of the net $\CA$ associated with the inclusion $\N\subset\M$ (Def.~\ref{def:netA}), and quite simple in its own right (Prop.~\ref{prop:triviality-with-limits}). 

In Section~\ref{section:examples}, we introduce our examples. The starting point of our strategy is a second quantization half-sided modular inclusion $\N\subset\M$, described in terms of simpler ``one-particle'' data (irreducible standard pair) and corresponding Weyl operators. In Section~\ref{section:defdeformed} we perturb a set of generating Weyl operators in a particular and explicit manner depending on a deformation parameter $\kappa\geq0$. These deformed operators define half-sided inclusions $\N_\kappa\subset\M_\kappa$ that coincide with the original second quantized situation for $\kappa=0$. Our main result is Theorem~\ref{thm:main}, stating that for any $\kappa>0$, these inclusions have trivial relative commutant.

Our method is put into a wider context in Section~\ref{section:warping}. We explain how our initial inclusion can be viewed as arising from a free ``bulk QFT'' on $\Rl^2$ by restriction to a light ray. The deformed inclusion then amounts to carrying out a warped convolution deformation \cite{BuchholzLechnerSummers:2011} in $\Rl^2$ and restrict the resulting QFT back to the light ray. 

In Section~\ref{section:proofs} we prove Theorem~\ref{thm:main}. This is done by verifying the criterion developed in Section~\ref{section:algebrainfinity}, which amounts to controlling certain weak limits $\wlim\limits_{t\to-\infty}\sigma_t(S)$ of operators $S$ localized away from $0\in\Rl$ under the modular group $\sigma_t$ of $\M$ w.r.t. the vacuum vector. We proceed in two steps, first analyzing such limits for unbounded field operator polynomials (Theorem~\ref{thm:limit-XY}), and then for bounded operators (Theorem~\ref{thm:boundedlimit}). Our conclusions and an outlook are presented in Section~\ref{section:conclusion}.

\section{Three types of half-sided modular inclusions} \label{section:3cases}

\begin{definition}\label{def:HSM}
    A half-sided modular inclusion is an inclusion of two von Neumann algebras $\N\subset\M$ on a Hilbert space $\Hil$ with a vector $\Om$ that is cyclic and separating for both $\N$ and $\M$ and such that the modular group $\sigma_t=\ad\Delta^{it}$ of $(\M,\Om)$ acts on $\N$ according to $\sigma_t(\N)\subset\N$ for all $t\leq0$.
\end{definition}

\begin{definition}\label{def:BT}
    A one-dimensional Borchers triple $(\M,T,\Om)$ consists of a von Neumann algebra $\M$ on a Hilbert space $\Hil$, a strongly continuous unitary one-parameter group $T$ with positive generator $\Pst$, and a unit vector $\Om\in\Hil$ such that
    \begin{enumerate}
        \item $\Om$ is cyclic and separating for $\M$, and $T(x)\Om=\Om$ for all $x\in\Rl$,
        \item $T(x)\M T(x)^{-1}\subset\M$ for $x\geq0$.
    \end{enumerate}
\end{definition}

We recall that in the situation of a one-dimensional Borchers triple $(\M,T,\Om)$, Borchers' Theorem \cite{Borchers:1992} asserts that 
\begin{align}\label{HJB}
    \Delta^{it}T(x)\Delta^{-it}=T(e^{-2\pi t}x),\quad JT(x)J=T(-x),\qquad t,x\in\Rl,
\end{align}
where $J,\Delta$ are the modular data of $(\M,\Om)$. Thus $T$ extends to a (anti)unitary representation $U$ of the affine group $G$ of $\Rl$ (the ``$ax+b$ group'').

As a consequence of \eqref{HJB}, the inclusion $\N:=T(1)\M T(-1)\subset\M$ is half-sided modular. Conversely, given a half-sided modular inclusion $(\N\subset\M,\Om)$, there exists a strongly continuous unitary one-parameter group $T$ such that $(\M,T,\Om)$ is a one-dimensional Borchers triple, and $\N=T(1)\M T(-1)$ \cite{Wiesbrock:1993-1,ArakiZsido:2005}. We will therefore use the terms half-sided modular inclusion and one-dimensional Borchers triple interchangeably to describe this structure. In addition to $\N,\M,\Om,T,U$, we will constantly use the notations
\begin{align}
    \sigma_t=\ad\Delta^{it},\qquad \alpha_x=\ad T(x),\qquad \om=\langle\Om,\,\cdot\,\Om\rangle.
\end{align}
To avoid confusion, we also emphasize that the modular data $\Delta,J$ will always be the ones of $(\M,\Om)$, so that no further subscript like $\Delta_\M$ will be necessary.

We will ask in addition that the vacuum vector $\Om$ is uniquely characterized (up to multiples) by the condition that it is invariant under $T$, namely $\ker\Pst=\Cl\Om$. We will refer to this condition as ``uniqueness of the vacuum''. In this case, $\M$ is a factor, more specifically a type III${}_1$ factor \cite{Longo:1979,Wiesbrock:1993-1} (disregarding the trivial case $\dim\Hil=1$).

\bigskip 

The physical interpretation of these data is as follows: $\Om$ denotes the vacuum state of a chiral half of a conformal QFT on a light ray, and $T$ corresponds to translation along that light ray. The algebra $\M$ is generated by all observables localized in the right half ray $\Rl_+
$ as expressed by the half-sided invariance under translations in Def.~\ref{def:BT}. Similarly, the commutant $\M'$ contains the observables localized in the left half ray $\Rl_-$. Observables localized in bounded intervals $I\subset\Rl$ (We denote the set of all open bounded intervals $I\subset\Rl$ by $\CI$) can be extracted from this setting by forming appropriate relative commutants.

\begin{definition}\label{def:netA}
    Let $(\N\subset\M,\Om)$ be a half-sided modular inclusion with associated representation $T$ of $\Rl$. For bounded intervals $(a,b)\subset\Rl$, we define 
    \begin{align}
        \CA(a,b):=\alpha_a(\M)\cap\alpha_b(\M)',
    \end{align}
    and call the resulting map $\CI\ni I\mapsto\CA(I)\subset\B(\Hil)$ from bounded intervals of $\Rl$ to von Neumann algebras in $\B(\Hil)$ {\em the local net associated with} $(\N\subset\M,\Om)$.
\end{definition}

As is well known, the local net $I\mapsto\CA(I)$ associated with any half-sided modular inclusion has many physically relevant properties: i) It is a net, i.e. inclusion preserving, ii)~it is local in the sense that $\CA(I_1)$ and $\CA(I_2)$ commute if the intervals $I_1$ and $I_2$ are disjoint, and iii)~it is covariant in the sense that 
\begin{align}
    U(g)\CA(I)U(g)^{-1} = \CA(gI),\qquad g\in G,\,I\in\CI\,.
\end{align}
The data $(\CA(I))_{I\in\Hil},U,\Om$ are therefore quite close to satisfying all axioms of chiral conformal field theory \cite{Longo:2008}. However, while it is clear from the definition that $\Om$ is separating for any $\CA(I)$, $I\in\CI$, it is in general not cyclic. Since the selfadjoint elements of $\CA(I)$ are interpreted as the observables measurable in $I$, this is a point that needs to be investigated carefully.

As a rough notion of ``size'' of the local observable content, we define the {\em local subspace} as 
\begin{align}\label{hilloc}
 \Hilloc := \overline{\CA(I)\Om} \subset\Hil,\qquad I\in\CI.
\end{align}
The following statement is well known but of crucial importance here, so we briefly recall its (Reeh-Schlieder type) proof (see, for example \cite[Lemma 5.1]{BostelmannLechnerMorsella:2011}).

\begin{lemma}
    The subspace \eqref{hilloc} is independent of the choice of interval $I\in\CI$.
\end{lemma}
\begin{proof}
    We will show that for an arbitrary interval $I\in\CI$, a vector $\Psi\perp\CA(I)\Om$ satisfies $\Psi\perp\CA(\tilde I)\Om$ for {\em all} $\tilde I\in\CI$; this implies the claim.
    
    Having fixed $I$ and $\Psi$ in this manner, pick $I_0\subset I$ such that $\overline{I_0}\subset I$, so that $f(x):=\langle\Psi,T(x)A\Om\rangle$, $A\in\CA(I_0)$, vanishes for small enough $|x|$. Since the generator of $T$ is positive, $f$ analytically continues to the complex upper half plane. Hence, by the edge of the wedge theorem, $f(x)=0$ for all $x$, i.e. $\Psi\perp\CA(I_0+x)\Om$ for all $x\in\Rl$. Choosing $x$ such that $I_0+x=(0,a)$ for some $a>0$, we next consider $g(t):=\langle\Psi,\Delta^{it}B\Om\rangle$, $B\in\CA(0,a)$. As the modular group acts by dilations, we have $\sigma_t(\CA(0,a))=\CA(0,e^{-2\pi t}a)\subset\CA(0,a)$ for $t>0$; this yields $g(t)=0$ for $t>0$. On the other hand, $\CA(0,a)\subset\M$ implies that the vector $B\Om$ lies in the domain of $\Delta^{1/2}$. Hence $g$ has an analytic extension to the strip $-\frac{1}{2}<\Im t<0$, and we may again apply the edge of the wedge theorem to conclude $g(t)=0$ for all $t$. 
    
    As any interval $\tilde I$ arises from $I_0$ by dilation and translation, the claim follows.
\end{proof}

We denote the orthogonal projection onto $\Hilloc$ by~$\Ploc$. Note that we always have $\Om\in\Hilloc$, so the projection $\POm$ onto $\Cl\Om$ is always a subprojection of $\Ploc$. Each half-sided inclusion therefore belongs to exactly one of the following three cases:

\begin{description}
    \item[\Std] {\em The standard case\footnote{This terminology is justified by the fact that in this case, $\Om$ is a standard (cyclic and separating) vector for the relative commutant $\N'\cap\M=\CA(0,1)$, but it is not ``standard'' in the sense of representing the generic situation.}:} $\Ploc=1$.
    \item[\Int] {\em The intermediate case:} $\POm\lneq\Ploc\lneq1$. 
    \item[\Sing] {\em The singular case:} $\Ploc=\POm$. 
\end{description}

As the algebra $\CA((0,1))$ of the unit interval $(0,1)$ coincides with the relative commutant $\N'\cap\M$ of $\N\subset\M$, and $\Hilloc=\overline{\CA((0,1))\Om}$, the three cases can also be described in terms of $\N'\cap\M$. Namely, the standard case \Std is equivalent to $\N'\cap\M$ having $\Om$ as a cyclic vector, the intermediate case \Int is equivalent to $\N'\cap\M\neq\Cl1$ being non-trivial but not having $\Om$ as a cyclic vector, and the singular case \Sing is equivalent to $\N'\cap\M=\Cl1$ being trivial, i.e. the inclusion $\N\subset\M$ being singular (irreducible subfactor).

\begin{itemize}
    \item[\Std] From the point of view of applications to conformal quantum field theory, the standard case \Std is the desired situation. In this case, $U$ extends to a strongly continuous representation of the Möbius group PSL$(2,\Rl)$ and $\CA$ can be formulated as a net on the unit circle which transforms covariantly under this Möbius action. Indeed, Guido, Longo, and Wiesbrock have shown that there is a bijection between (isomorphism classes of) standard (case \Std) half-sided modular inclusions and strongly additive local conformal nets on the circle \cite[Cor.~1.9]{GuidoLongoWiesbrock:1998}. 

    Many examples of case \Std are known, for instance half-sided inclusions generated by chiral Wightman fields and all half-sided inclusions arising from second quantization of standard pairs (see Sect.~\ref{section:examples}).
 
    \item[\Int] The intermediate case \Int is closely related to the standard case \Std. In fact, since the algebras $\CA(I)$, $I\in\CI$, and the $G$-representation $U$ leave $\Hilloc$ invariant, one can in case \Int restrict all data to $\Hilloc$ and then obtain a standard one-dimensional Borchers triple
    \begin{align}
        \M_{\rm loc} := \bigvee_{I\in\CI,I\subset\Rl_+}\CA(I)|_{\Hilloc} \subset \B(\Hilloc),\qquad T_{\rm loc} := T|_{\Hilloc}, \qquad \Om.
    \end{align}
    The above comments about case \Std apply to the triple $(\M_{\rm loc}, T_{\rm loc},\Om)$ on~$\Hilloc$ without changes.
    
    Several examples of case \Int are known, see for example \cite{BostelmannLechnerMorsella:2011}. Further examples can be constructed by taking tensor products of case \Std and \Sing.
    
    \item[\Sing] From the point of view of quantum field theory, the singular case \Sing is pathological as it describes a net without any (non-trivial) local observables. From the point of view of operator algebras, these are the irreducible subfactors coming from half-sided modular inclusions.
    
    Only very recently a first example of case \Sing has been found: Longo, Tanimoto, and Ueda \cite{LongoTanimotoUeda:2019} use methods of free probability to construct singular half-sided modular inclusions.    
\end{itemize}

Summarizing the above discussion, it is known that all three cases occur. However, there exist also many undecided examples for which it is currently unclear to which of the cases \Std,\Int,\Sing they belong. For example, the short distance scaling limits of integrable quantum field theories considered in \cite{BostelmannLechnerMorsella:2011} fall into two infinite families (+) and ($-$) depending on the sign of the limit of the scattering function at infinite rapidity transfer. Examples from family (+) could be of case \Std, \Int, or \Sing, and examples from family ($-$) could be of case \Int or \Sing, but in all but two special examples, the exact case is not known. In field theories with several particle species and internal degrees of freedom the spectrum of possibilities is even much wider \cite{Scotford:}.

Given the similarity of cases \Std and \Int, it is therefore an important question to develop tools that allow to decide whether a given half-sided inclusion is singular or not. 

\medskip 

As mentioned before, half-sided modular inclusions with unique vacuum are examples of type III${}_1$ subfactors which do not allow for the application of well-known techniques for investigating their relative commutant. For instance, such a subfactor does not allow for a normal conditional expectation $\M\to\N$, and no meaningful notion of index exists. These inclusions can also never be split \cite{DoplicherLongo:1984} so that the ideas relating to modular compactness/nuclearity \cite{BuchholzDAntoniLongo:1990,BuchholzLechner:2004} do not apply here. This lack of tools explains the many undecided examples.

\medskip

\section{Half-sided inclusions and the algebra at infinity}\label{section:algebrainfinity}

Our condition for detecting the singular case $\Sing$ is based on the notion of an algebra at infinity, familiar from the study of quasilocal algebras \cite{BratteliRobinson:1987}. For our purposes, the right definition is the following.

\begin{definition}
    Let $(\N\subset\M,\Om)$ be a half-sided modular inclusion and $\CA$ its associated local net. Its {\em algebra at infinity} is the von Neumann algebra 
    \begin{align}
        \scrX
        :=
        \bigcap_{I\in\CI}\CA(I)'.
    \end{align}
\end{definition}

Equivalent formulations of the algebra at infinity are 
\begin{align}
   \scrX
    =
    \bigcap_{x>0}\left(\alpha_x(\M)\vee\alpha_{-x}(\M')\right)
     =
    \bigcap_{t<0}\sigma_t(\N\vee J\N J)
    .
\end{align}
Indeed, since $\alpha_x(\M)\vee\alpha_{-x}(\M')=\CA(I_x)'$, where $I_x=(-x,x)$, $x>0$, we have the inclusion $\subset$ at the first equality sign, and the opposite inclusion follows by isotony. By Borchers' Theorem, we have $\sigma_t(\N\vee J\N J)=\sigma_t(\alpha_1(\M)\vee\alpha_{-1}(\M'))=\alpha_{e^{-2\pi t}}(\M)\vee\alpha_{-e^{-2\pi t}}(\M')$, which implies the last formula for $\scrX$.

It is instructive to compare $\scrX$ with the algebras at left/right infinity, namely 
\begin{align}
    \scrX_{\rm right} = \bigcap_{t<0}\sigma_t(\N),\qquad
    \scrX_{\rm left} = \bigcap_{t<0}\sigma_t(J\N J).
\end{align}
If the vacuum is unique, one always has $ \scrX_{\rm right}=\scrX_{\rm left}=\Cl1$ because the translations act trivial on them \cite{Longo:1984}. The algebra $\scrX$ at (left {\em and} right) infinity can however be large:

\begin{proposition}\label{prop:triviality-criterion}
    Let $(\N\subset\M,\Om)$ be a half-sided modular inclusion with algebra at infinity $\scrX$. The following are equivalent:
    \begin{enumerate}
        \item\label{triv1} $\N\subset\M$ is singular (case \Sing), i.e. $\Ploc=\POm$.
        \item\label{triv2} $\scrX=\B(\Hil)$.
        \item\label{triv3} $\POm\in\scrX$.
    \end{enumerate}
\end{proposition}
\begin{proof}
    \ref{triv1}$\Rightarrow$\ref{triv2} In the singular case \Sing, we have $\CA(I)'=(\Cl1)'=\B(\Hil)$ for all $I\in\CI$ and therefore $\scrX=\B(\Hil)$. \ref{triv2}$\Rightarrow$\ref{triv3} is trivial. For \ref{triv3}$\Rightarrow$\ref{triv1} we take $A\in\CA((0,1))=\N'\cap\M$. Then $\POm\in\CA((0,1))'$, i.e. $A\Om=A\POm\Om=\POm A\Om=\om(A)\Om$. Since $\Om$ separates $\CA((0,1))$, we get $A=\om(A)1$ and hence \ref{triv1}.
\end{proof}

Every element $X\in\scrX$ can be understood as an obstruction to the existence of local observables $A\in\CA(I)$ because $A$ has to commute with $X$.

\medskip 

We next discuss how to obtain elements of $\scrX$. The idea is to consider an operator $A\in\N\vee J\N J$, i.e. localized in $(-\infty,-1]\cup[1,\infty)$, and scale it with the modular group to operators $\sigma_t(A)$ localized in $(-\infty,-e^{-2\pi t}]\cup[e^{-2\pi t},\infty)$. In the limit $t\to-\infty$ we then obtain elements of $\scrX$.

Before we explain this further, let us recall that the uniqueness of the vacuum implies the weak limits
\begin{align}\label{eq:weaklimits}
    \wlim_{x\to\pm\infty}T(x)=\POm,\qquad \wlim_{t\to\pm\infty}\Delta^{it}=\POm.
\end{align}
This is the case because as a consequence of the representation theory of $G$, the restrictions of the selfadjoint generators $\Pst$ and $\log\Delta$ to the orthogonal complement $\POm^\perp\Hil$ of $\Cl\Om$ have purely absolutely continuous spectrum, see e.g. \cite{Longo:2008}. The limits \eqref{eq:weaklimits} then follow by application of the Riemann-Lebesgue Lemma.

\begin{lemma}\label{lemma:limits}
    Let $A\in\N\vee J\N J$ be such that $\sigma_t(A)$ converges weakly to $L\in\B(\Hil)$ as $t\to-\infty$. Then a) $L\in\scrX$, b) $[L,\Delta^{it}]=0$ for all $t\in\Rl$, c) $L\Om=\om(A)\Om$.
\end{lemma}
\begin{proof}
    a) At fixed $t\in\Rl$, we have $\sigma_{t}(A)\in\sigma_{t}(\N\vee J\N J)=\alpha_{e^{-2\pi t}}(\M)\vee\alpha_{-e^{-2\pi t}}(\M')=:\M_{t}$. Since $\M_{t}\subset\M_{s}$ for $t<s$, the weak limit $L$ is an element of $\bigcap_{t<0}\M_{t}=\scrX$. 
    
    b) By assumption, $\sigma_t(A)\to L$ weakly as $t\to-\infty$. Thus also $\sigma_{t+s}(A)\to L$ weakly. But $\sigma_{t+s}(A)=\Delta^{is}\sigma_t(A)\Delta^{-is}\to\sigma_s(L)$ as $t\to-\infty$, so $\sigma_s(L)=L$ follows. 
    
    c) By assumption, $\sigma_t(A)\to L$ weakly as $t\to-\infty$, and hence $\Delta^{it}A\Om=\sigma_t(A)\Om\to L\Om$ weakly. But on the other hand, $\Delta^{it}A\Om\to\POm A\Om=\om(A)\Om$ by \eqref{eq:weaklimits}. Thus $L\Om=\om(A)\Om$.
\end{proof}

Recalling that norm closed and bounded subsets of $\B(\Hil)$ are weakly compact, $\sigma_t(A)$, $t\to-\infty$, has weak limit points. So the above lemma could be reformulated in terms of weak limit points to avoid the assumption of existence of the weak limit. However, for us the above form will be sufficient.

Our criterion for singular inclusions now follows by combining Prop.~\ref{prop:triviality-criterion} and Lemma~\ref{lemma:limits}.

\begin{proposition}\label{prop:triviality-with-limits}
    Let $(\N\subset\M,\Om)$ be a half-sided modular inclusion with unique vacuum. The following are equivalent:
    \begin{enumerate}
        \item\label{trivlim1} $\N\subset\M$ is singular (case \Sing).
        \item\label{trivlim2} There exists $T\in\N\vee J\N J$ such that $T\Om\neq0$ and $\slim\limits_{t\to-\infty}T\Delta^{-it}$ exists.
        \item\label{trivlim3} There exist $S\in\N\vee J\N J$ such that 
        \begin{align}
            \wlim_{t\to-\infty}\sigma_t(S)=\POm.
        \end{align}
    \end{enumerate}
\end{proposition}
\begin{proof}
    \ref{trivlim1}$\Rightarrow$\ref{trivlim2} If $\N\subset\M$ is singular, we have $\N\vee J\N J=\B(\Hil)$ and hence we may choose $T=\POm$ to satisfy the assumptions in \ref{trivlim2}.
    
    \ref{trivlim2}$\Rightarrow$\ref{trivlim3} Let $L\in\B(\Hil)$ denote the strong limit of $T\Delta^{-it}$ as $t\to-\infty$. Then in particular $T\Delta^{-it}\to L$ weakly. But $\Delta^{-it}\to\POm$ weakly, so we conclude $L=T\POm$. This implies $\sigma_t(T^*T)=(T\Delta^{-it})^*\cdot T\Delta^{it}\to L^*L=\POm T^*T\POm=\|T\Om\|^2\POm$ weakly. Since $T\Om\neq0$ by assumption, we may consider $S=T^*T/\|T\Om\|^2$ which satisfies the assumption in \ref{trivlim3}.
     
    \ref{trivlim3}$\Rightarrow$\ref{trivlim1} The weak limit $\wlim\limits_{t\to-\infty}\sigma_t(S)=\POm$ lies in $\scrX$ by Lemma~\ref{lemma:limits}, which implies \ref{trivlim1} by Prop.~\ref{prop:triviality-criterion}.
\end{proof}

In the next section, we will give examples of half-sided modular inclusions which we will demonstrate to be singular by verifying Prop.~\ref{prop:triviality-with-limits}~\ref{trivlim3}.

\section{Second quantization inclusions and their deformations}\label{section:examples}

\subsection{Definition of the deformed and undeformed inclusions}\label{section:defdeformed}

In this section we describe the particular half-sided modular inclusions that we will investigate. Our starting point is a standard subspace version of a Borchers triple, namely a {\em non-degenerate standard pair} $(T_1,H)$ \cite{LongoWitten:2010}.

\begin{definition}
    A standard pair $(T_1,H)$ on a complex Hilbert space $\Hil_1$ consists of 
    \begin{enumerate}
        \item a closed real subspace $H\subset\Hil_1$ that is {\em standard}, namely cyclic in the sense that $H+iH$ is dense in $\Hil$, and separating in the sense $H\cap iH=\{0\}$,
        \item a strongly continuous unitary one-parameter group $T_1(x)=e^{ix\Pst_1}$ with positive non-singular generator $\Pst_1>0$, $\ker \Pst_1 =\{0\}$,
    \end{enumerate}
    such that
    \begin{align}\label{eq:T1halfsided}
        T_1(x)H\subset H,\qquad x\geq0.
    \end{align}
\end{definition}
As is well known, a standard pair gives rise to a half-sided modular inclusion / one-dimensional Borchers triple by second quantization. Namely, consider the Bose Fock space $\Hil$ over $\Hil_1$ with its canonical vacuum vector $\Om$ and Weyl unitaries $V(h)$, $h\in\Hil_1$. Then the von Neumann algebra 
\begin{align}\label{eq:MH}
    \M(H) := \{V(h)\,:\,h\in H\}''
    \subset
    \B(\Hil),
\end{align}
the second quantized representation $T(x)=e^{ix\Pst}$, $\Pst=\Gamma(\Pst_1)$, and the Fock vacuum $\Om$ form a one-dimensional Borchers triple $(\M(H),T,\Om)$ with unique vacuum. It is known that such triples always belong to the standard case \Std (see \cite[Prop.~2.3]{LechnerLongo:2014}, related to \cite[Thm.~4.5]{BrunettiGuidoLongo:2002}). The modular data $\Delta,J$ of $(\M(H),\Om)$ arise from the modular data $\Delta_1,J_1$ of $H$ (defined by polar decomposition of $S_1:H+iH\to H+iH,h+ih\mapsto h-ih$) by second quantization. The commutant of $\M(H)$ is $\M(H)'=\M(H')$, where $H'=\{\psi\in\Hil_1\,:\,\Im\langle\psi,h\rangle=0\,\forall h\in H\}=J_1H$ is the symplectic complement of $H$ \cite{LeylandsRobertsTestard:1978}.

\medskip 

A standard pair $(T_1,H)$ is called irreducible iff the $G$-representation generated by the one-parameter groups $T_1(x)$ and $\Delta^{it}$ is irreducible. In this case, the pair is unique up to unitary equivalence and can be presented in the form \cite[p.~40, case (i)]{LechnerLongo:2014}
\begin{align}\label{eq:rapidityrep}
    \Hil_1&=L^2(\Rl,d\te),\qquad (T_1(x)\psi)(\te)=e^{ixe^\te}\psi(\te),
\end{align}
with standard subspace \cite[Lemma~4.1]{LechnerLongo:2014}
\begin{align}\label{def:H}
    H &= \{\psi\in\Hardy^2(\Strip_\pi)\,:\overline{\psi(\te+i\pi)}=\psi(\te)\,\text{ a.e.}\}.
\end{align}
Here $\Hardy^2(\Strip_\pi)$ denotes the Hardy space on the strip $\Strip_\pi=\{\zeta\in\Cl\,:\,0<\Im\zeta<\pi\}$, namely all those $\psi\in L^2(\Rl)$ that are boundary values of analytic functions $\psi:\Strip_\pi\to\Cl$ with $\sup_{0<\la<\pi}\int_\Rl|\psi(\te+i\la)|^2d\te<\infty$. The modular data of $H$ are 
\begin{align}
    (\Delta_1^{it}\psi)(\te)=\psi(\te-2\pi t),\qquad (J_1\psi)(\te)=\overline{\psi(\te)}.
\end{align}
It should be noted that elements in $H$ can be obtained from real functions on $\Rl_+$ by Fourier transform, namely the map 
\begin{align}
 \Ss(\Rl)\to\Hil_1,\qquad f\mapsto \hat{f},\qquad \hat{f}(\te):=\int_\Rl f'(x)e^{ixe^\te}dx
\end{align}
carries $C_{c,\Rl}^\infty(\Rl_+)$ into a dense subspace of $H$,
\begin{align}\label{eq:H1d}
    H = \{\hat f\,:\, f\in C_{c,\Rl}^\infty(\Rl_+)\}^-.
\end{align}

\bigskip

We now explain the structure of the deformed versions of $(\M,T,\Om)$ which will depend on a real deformation parameter $\kappa$. To that end, we introduce explicit generating operators on Fock space $\Hil$. In the following, $\D\subset\Hil$ will always denote the dense subspace of finite particle number, and for $\Psi\in\Hil$, its $n$-particle component will be denoted $\Psi_n$, $n\in\Nl$.

Given $\kappa\in\Rl$ and $\psi\in\Hil_1$ we define an annihilation type operator by \cite{GrosseLechner:2007}
\begin{align}\label{def:ak}
    a_\kappa(\psi)&:\D\to\D\\
   (a_\kappa(\psi)\Psi)_n(\te_1,\ldots,\te_n)
    &=
    \sqrt{n+1}\int_\Rl d\te\,\overline{\psi(\te)}\prod_{k=1}^ne^{i\kappa\sinh(\te-\te_k)}\Psi_{n+1}(\te,\te_1,\ldots,\te_n)
    . \nonumber
\end{align}
For $\kappa=0$, the operator $a_0(\psi)$ reduces to the familiar Bose annihilation operator. The reason for the exponential factors $e^{i\kappa\sinh(\te-\te_k)}$ in \eqref{def:ak} will be explained below. 

We then define a field operator as
\begin{align}
    \varphi_\kappa(\xi)
    =
    a_\kappa(\xi)^*+a_\kappa(S_1\xi)
    ,\qquad 
    \xi\in H+iH,
\end{align}
which is a well-defined closable operator on $\D$ depending linearly on $\xi$. Let us recall the following facts \cite {Lechner:2012}:
\begin{itemize}
    \item $\varphi_\kappa(h)$ is essentially selfadjoint on $\D$ for $h\in H$. (The selfadjoint closure will be denoted by the same symbol.)
    \item For $\kappa\geq0$, we have $[e^{i\varphi_\kappa(h)},e^{i\varphi_{-\kappa}(h')}]=0$ for all $h\in H$, $h'\in H'$.
    \item $U(x,t)\varphi_\kappa(\xi)U(x,t)^{-1}=\varphi_\kappa(U_1(x,t)\xi)$ for any $x\geq0$ and any $t\in\Rl$.
    \item $J\varphi_\kappa(\xi)J=\varphi_{-\kappa}(J_1\xi)$.
    \item $\Om$ is cyclic for the algebra of all polynomials in $\varphi_\kappa(h)$, $h\in H$.
\end{itemize}

These structures are summarized in the following proposition.

\begin{proposition}\label{defprop}
    Let $\kappa\geq0$ and 
    \begin{align}\label{Mkappa}
        \M_\kappa := \{e^{i\varphi_\kappa(h)}\,:\,h\in H\}'' \subset \B(\Hil). 
    \end{align}
    Then $(\M_\kappa,T,\Om)$ is a one-dimensional Borchers triple with unique vacuum on the Bose Fock space $\Hil$ over $\Hil_1$. For $\kappa=0$ we have $\M_0=\M(H)$, the second quantization of the irreducible standard pair $(H,T_1)$.
\end{proposition}

In the following, we will use a subscript $\kappa$ to refer to the ``deformed'' von Neumann algebra $\M_\kappa$, its subalgebra $\N_\kappa=T(1)\M_\kappa T(-1)$, algebra at infinity $\scrX_\kappa$, local subspace ${\Hilloc}_\kappa$, etc. Note, however, that independently of $\kappa$, we are always working on the same Hilbert spaces $\Hil_1$ and $\Hil$, using the same vacuum vector $\Om$ and the same standard subspace $H$. 

By definition, the translations $T(x)$ are the same for all $(\N_\kappa\subset\M_\kappa,\Om)$, $\kappa\geq0$. We remark that also the modular data of $(\M_\kappa,\Om)$ are independent of~$\kappa$, and that the commutant of $\M_\kappa$ is given by \cite{BuchholzLechnerSummers:2011}
\begin{align}\label{MK'}
    {\M_\kappa}'=\{e^{i\varphi_{-\kappa}(h')}\,:\,h'\in H'\}'' = {\M'}_{-\kappa}.
\end{align}

Our main result is the following theorem.

\begin{theorem}\label{thm:main}
    Consider the family of one-dimensional Borchers triples $(\M_\kappa,T,\Om)$, $\kappa\geq0$, defined in \eqref{Mkappa}. For $\kappa=0$, this inclusion is standard (case \Std), and for any $\kappa>0$, this inclusion is singular (case \Sing).
\end{theorem}

This results highlights that despite the properties that the triples $(\M_\kappa,T,\Om)$ share for all $\kappa\geq0$, there are also essential differences between the undeformed $(\kappa=0)$ and deformed ($\kappa>0$) cases. Not only are the von Neumann algebras~$\M_\kappa$ {\em not} of second quantization form for $\kappa>0$, but their local subspaces~${\Hilloc}_\kappa$ and relative commutants ${\N_\kappa}'\cap\M_\kappa$ depends on $\kappa$ in a ``discontinuous'' manner. Whereas the triple for $\kappa=0$ describes a (chiral half of) a local QFT, the local observable content becomes trivial for $\kappa>0$.

It is also interesting to consider this situation from the point of view of the representations of the affine and Möbius groups. In the situation encountered here, the modular data $\Delta,J$ of $(\M_\kappa,\Om)$ are independent of the deformation parameter $\kappa$. So the representation $U$ of $G$ (``$ax+b$ group'') extends to a representation $\hat{U}_0$ of the Möbius group under which the net $\CA_0$ associated with $(\N_0\subset\M_0,\Om)$ transforms covariantly; this extension can be realized with the modular data of $\CA_0(I)$, $I\in\CI$. But the net associated with $(\N_\kappa\subset\M_\kappa,\Om)$, $\kappa>0$, is trivial, in particular $\hat U_0(g)$ does not map $\M_\kappa$ onto $\CA_\kappa((0,1))$ for $g(x)=\frac{1}{x+1}$. Hence, contrary to a statement sometimes found in the literature, extension of $G$ to a Möbius group representation is not sufficient for the existence of local observables -- one also has to ensure that the extended representation acts correctly on the net. 

This circle of ideas might connect to ongoing research relating modular theory of standard subspaces and representation theory \cite{NeebOlafsson:2017,NeebOlafsson:2021}.

\subsection{Two-dimensional nets and warped convolution}\label{section:warping}

In this subsection we describe the deformations from a different perspective that highlights their structural properties and explains the formula \eqref{def:ak}. The main point is that we may view the one-dimensional Borchers triple $(\M(H),T,\Om)$ also as a {\em two}-dimensional one. This point of view will give us sufficient room to work with a deformation scheme (warped convolution) that requires a suitable action of $\Rl^2$, see also \cite{Tanimoto:2011-1,LechnerSchlemmerTanimoto:2013} for related constructions.

\medskip 

By definition, every one-dimensional Borchers triple $(\M,T,\Om)$ comes with a representation $T$ of the one-dimensional translation group that acts according to $T(x)\M T(-x)\subset\M$ for $x\geq0$ (Def.~\ref{def:BT}). In comparison, a {\em two-dimensional Borchers triple} $(\M,\underline{T},\Om)$ is defined as a von Neumann algebra $\M$ with standard vector $\Om$ and a unitary representation $\underline{T}$ of the two-dimensional translation group, such that $\Om$ is invariant under $\underline{T}$,
\begin{align}
 \underline{T}(x,y)\M \underline{T}(x,y)^{-1} \subset \M,\qquad x\geq0,y\leq0,
\end{align}
and both one-parameter groups $\underline{T}(x,0)$ and $\underline{T}(0,y)$ have positive generators. 

The parameters $x,y$ should be thought of as the light ray coordinates of vectors $\xi\in\Rl^2$, namely $x=\xi_-$, $y=\xi_+$ with $\xi_\pm=\frac{1}{2}(\xi_0\pm\xi_1)$, where $\xi_0$ is the temporal and $\xi_1$ the spatial coordinate of $\xi$. Then $\underline{T}$ is a positive energy representation and the inequalities $x>0,y<0$ describe the right wedge in $\Rl^2$, namely $W_R=\{\xi\in\Rl^2\,:\,\pm\xi_\pm>0\}$.

\begin{center}
    \includegraphics[width=50mm]{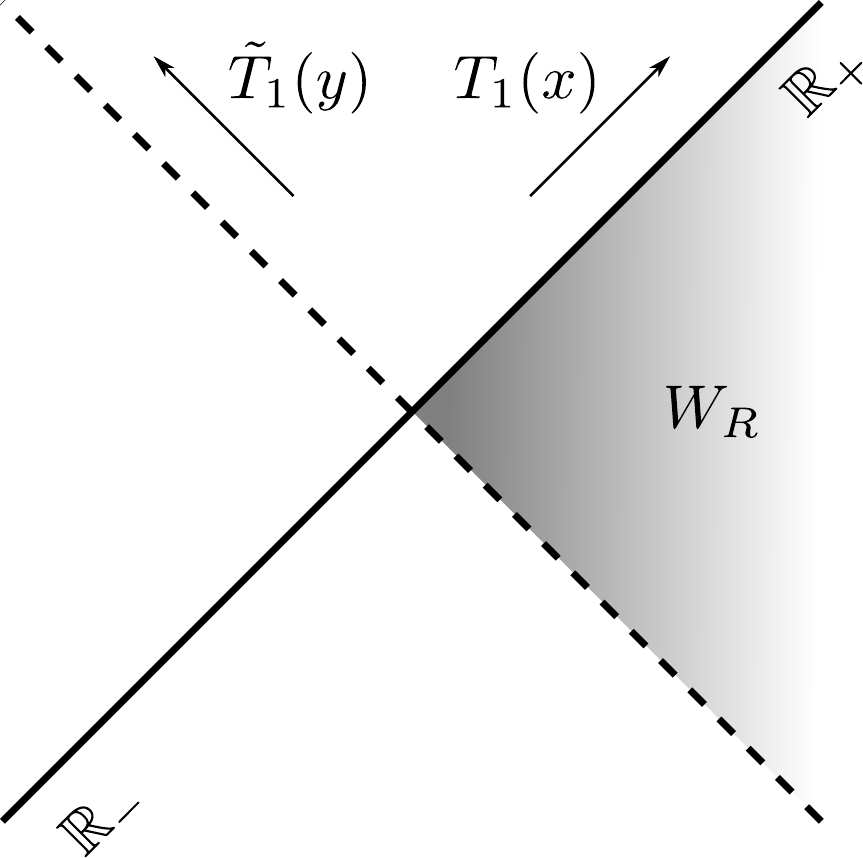}
\end{center}

Clearly, any two-dimensional Borchers triple is also a one-dimensional one by setting $T(x):=\underline{T}(x,0)$. Conversely, a one-dimensional Borchers triple is also a two-dimensional one by making the trivial choice $\underline{T}(x,y):=T(x)$. 

For one-dimensional Borchers triples $(\M,T,\Om)$ arising from second quantization of standard pairs $(H,T_1)$, there exist however always non-trivial extensions of $T$ to a representation $\underline{T}$ of $\Rl^2$. In \cite[Cor.~2.5]{LongoWitten:2010}, all strongly continuous one-parameter groups $\tilde T_1(y)$ commuting with $T_1(x)$ and satisfying $\tilde T_1(y)H\subset H$ for $y\leq0$ have been classified. These one-parameter groups include in particular\footnote{This is also easily checked directly in case of the standard pair \eqref{def:H}.}
\begin{align}\label{T1t}
 \tilde T_1(y) = e^{im^2y\Pst_1^{-1}},
\end{align}
where $\Pst_1$ is the generator of $T_1$, and $m^2\geq0$ a parameter with the interpretation of mass (squared). We remark that \eqref{T1t} are the only possibilities with positive generator in case $(H,T_1)$ is irreducible: In fact, assuming that the generator $\tilde \Pst_1$ of $\tilde T_1$ is positive, Borchers' commutation relations imply that the product $\Pst_1\tilde \Pst_1$ commutes with $U$ and is thus a multiple of the identity by irreducibility. 

Since the value of $m^2>0$ in \eqref{T1t} plays no role for our investigations, we will set it to $m^2=1$. 

In the case of the presentation \eqref{def:H} of the irreducible standard pair, $\tilde T_1$ takes the form 
\begin{align}\label{eq:2Ps}
    (\tilde T_1(y)\psi)(\te) = e^{iye^{-\te}}\psi(\te),
\end{align}
and we can summarize our discussion as follows.

\begin{lemma}\label{lemma:2dinterpretation}\leavevmode
    \begin{enumerate}
     \item\label{item:rep1} The unitary $\Rl^2$-representation on $L^2(\Rl,d\te)$ given by\footnote{Here, ``$\cdot$'' denotes the Minkowski inner product.}
     \begin{align}\label{eq:2drep}
        (\underline{T}_1(\xi)\psi)(\te) = e^{ip(\te)\cdot \xi}\psi(\te),\qquad \xi\in\Rl^2,\qquad p(\te):=(\cosh\te,\sinh\te)
    \end{align}
    has positive energy and satisfies 
    \begin{align}\label{eq:2dincl}
        \underline{T}_1(\xi)H\subset H,\qquad \xi\in W_R.
    \end{align}
    It is related to $T_1$ by $\underline{T}_1(\xi)=T_1(\xi_-)\tilde T_1(\xi_+)$, where $\xi_\pm=\frac{1}{2}(\xi_0\pm\xi_1)$.
    \item\label{item:rep3} The second quantization von Neumann algebra $\M(H)$ \eqref{eq:MH}, the second quantization $\underline{T}$ of $\underline{T}_1$, and the Fock vacuum $\Om$ form a two-dimensional Borchers triple. 
    \end{enumerate}
\end{lemma}
\begin{proof}
    \ref{item:rep1} Formula \eqref{eq:2drep} follows from \eqref{eq:2Ps} by using light ray coordinates, and positivity is clear. \ref{item:rep3} follows immediately from \ref{item:rep1} by second quantization.
\end{proof}

In the setting of a two-dimensional Borchers triple, $\underline{T}(\xi)$, $\xi\in\Rl^2$ and $\Delta^{it}$, $t\in\Rl$, generate a representation of the Poincaré group $\PG$ on two-dimensional Minkowski space $\Rl^2$.

\bigskip 

Once we are in a situation of a two-dimensional Borchers triple $(\M,\underline{T},\Om)$, we are in position to apply the warped convolution deformation \cite{BuchholzLechnerSummers:2011}. This is a deformation procedure for two- (or higher-) dimensional Borchers triples based on a deformation of smooth operators $A\in\CB(\Hil)$. Here we call an operator $A\in\B(\Hil)$ {\em smooth} if the functions $x\mapsto \underline{T}(x)A \underline{T}(-x)$ are smooth in the norm topology of $\B(\Hil)$. A vector $\Psi\in\Hil$ is called {\em smooth} if $x\mapsto \underline{T}(x)\Psi$ is smooth in the norm topology of $\Hil$.

Choosing a real $(2\times 2)$ matrix $Q$ that is antisymmetric w.r.t. the Minkowski inner product, i.e. necessarily of the form $$Q=Q_\kappa=\begin{pmatrix}0&\kappa\\\kappa&0                                              \end{pmatrix},$$ for some $\kappa\in\Rl$, one then considers the integral expression (for smooth $A$ and~$\Psi$)
\begin{align}\label{warp}
    A_\kappa \Psi := (2\pi)^{-2}\int dp\,dx\,e^{-ip\cdot x}\underline{T}(Q_\kappa p)A\underline{T}(-Q_\kappa p)\,\underline{T}(x)\Psi.
\end{align}
Interpreted in an oscillatory sense, this integral defines a map $A_\kappa:\Hil^\infty\to\Hil^\infty$ that can be extended to a bounded operator, still denoted $A_\kappa$. It is easy to see that $A_0=A$ and $A_\kappa\Om=A\Om$; for further properties we refer to \cite{BuchholzLechnerSummers:2011}.

In the context of the Borchers triple $(\M,\underline{T},\Om)$, one then defines 
\begin{align}\label{Mk2}
    \M_\kappa := \{A_\kappa\,:\,A\in\M\,\text{ smooth}\}''.
\end{align}
Note that as $A\mapsto A_\kappa$ is not an algebra homomorphism, this algebra contains also elements that are not of the form $B_\kappa$, $B\in\M$.

By exploiting the positivity of the joint spectrum $S\subset\Rl^2$ of the generators of $\underline{T}$ one can then show: If $Q_\kappa S$ lies in the right wedge $W_R$, then $\Om$ is standard for $\M_\kappa$. This condition is equivalent to $\kappa\geq0$ and explains this restriction on the parameter in Prop.~\ref{defprop}. Furthermore, one can show that the commutant of $\M_\kappa$ is given by the opposite deformation of the commutant of $\M$ \eqref{MK'}. Since the action of the translations on deformed operators $A_\kappa$ is easy to control, it then follows that also $(\M_\kappa,\underline{T},\Om)$ is a Borchers triple \cite{BuchholzLechnerSummers:2011}. 

The initial deformation formula \eqref{def:ak} for the annihilation operator $a_\kappa(\xi)$ amounts to using the representation $\underline{T}$ in the warped convolution and checking $a_\kappa(\xi)=a_0(\xi)_\kappa$. The exponential factors involving $\sinh$ appear in \eqref{def:ak} because 
\begin{align}\label{pQp}
    p(\te)\cdot Q_\kappa p(\te') = \kappa\sinh(\te'-\te).
\end{align}

We also recall that elements of $H$ can be generated from test functions on $\Rl^2$, completing the two-dimensional picture. Namely, the maps 
\begin{align}
    \Ss(\Rl^2) \ni f\mapsto f^\pm\in\Hil_1,\qquad f^\pm(\te)=\int_{\Rl^2} f(\xi)e^{\pm ip(\te)\cdot\xi}d\xi
\end{align}
generates $H$ according to
\begin{align}\label{eq:H2d}
    H = \{f^+\,:\, f\in C_{c,\Rl}^\infty(W_R)\}^-.
\end{align}  

The operator-valued distribution $\phi_0$ on $\Rl^2$ defined by
\begin{align}
    \phi_0(f) := a_0^*(f^+)+a_0(\overline{f^-}),\qquad f\in\Ss(\Rl^2),
\end{align}
is nothing but the free scalar Klein Gordon field of unit mass on two-dimensional Minkowski space; it coincides with $\varphi_0(f^+)$ for $f\in C_{c,\Rl}(W_R)$. For any $\kappa\geq0$, the operators $\phi_\kappa(f)$, $f\in C_c^\infty(W_R)$, are affiliated with $\M_\kappa$.

\medskip 

To summarize, the definition \eqref{Mkappa} amounts to considering the free Klein-Gordon field in two dimensions, performing a warped convolution deformation with parameter $Q_\kappa$, and restricting the deformed quantum field theory to the light ray. For physical interpretations of such a light ray holography (in a somewhat different context), see for instance \cite{Schroer:2001-1}. It should also be noted that the deformed field theory (and related models, including models in higher dimensions) are interacting in the sense of having a non-trivial S-matrix; see \cite{GrosseLechner:2007} for two-particle scattering and \cite{Duell:2018,Duell:2019} for an analysis of $n$-particle scattering.

\subsection{Analysis of the deformed inclusions}\label{section:proofs}

This section is devoted to the proof of Theorem~\ref{thm:main} by verifying the weak limit criterion in Prop.~\ref{prop:triviality-with-limits}~\ref{trivlim3}, namely the existence of $S\in{\N_\kappa}\vee J\N_\kappa J$ such that $\sigma_t(S)\to\POm$ weakly as $t\to-\infty$. We will proceed in two steps: In a first step, we will consider field polynomials as an unbounded analogue of $S$, and prove a corresponding limit formula by investigating their correlation functions. In the second step, we will pass to bounded $S$. 

In preparation for the first step we recall Wightman type properties of the fields $\phi_\kappa$, $\kappa\in\Rl$, see \cite{GrosseLechner:2008} for details. 

We will use the vectors $\Psi(F)$, $F\in\Ss((\Rl^2)^n)$, that are defined by linear and continuous extension (kernel theorem) of $\Psi(f_1\ot\ldots\ot f_n):=\phi_0(f_1)\cdots\phi_0(f_n)\Om$, $f_1,\ldots,f_n\in\Ss(\Rl^2)$. On these vectors, $\phi_\kappa$ acts according to \cite{Soloviev:2008,GrosseLechner:2008}
\begin{align}\label{eq:phikappa}
    \phi_\kappa(g)\Psi(F)&=\Psi(g\ot_\kappa F),\qquad g\in\Ss(\Rl^2),\;F\in\Ss(\Rl^{2n}),
\end{align}
where the deformed tensor product $\ot_\kappa$ between two test functions $G\in\Ss(\Rl^{2m})$, $F\in\Ss(\Rl^{2n})$, $n,m\in\Nl$, is defined in Fourier space by
\begin{align}\label{eq:tensor}
    \widetilde{(G\ot_\kappa F)}(p_1,\ldots,p_m;q_1,\ldots,q_n)
    =
    e^{i\sum\limits_{l=1}^mp_l \cdot Q_\kappa \sum\limits_{r=1}^n q_r}
    \cdot
    \tilde{G}(p_1,\ldots,p_m)\tilde{F}(q_1,\ldots,q_n).
\end{align}
For $\kappa=0$, this reduces to the ordinary tensor product $\ot=\ot_0$. For every $\kappa\in\Rl$, the deformed tensor product $\ot_\kappa$ is an associative continuous product on the tensor algebra over $\Ss(\Rl^2)$ which is Poincar\'e-covariant in the sense 
\begin{align}
    \la^*(G\ot_\kappa F)
    =\la^* G\ot_\kappa \la^* F,\qquad \la\in\PG,\quad F\in\Ss(\Rl^{2n}),G\in\Ss(\Rl^{2m}).
\end{align}
While translation covariance is obvious from \eqref{eq:tensor}, covariance under boosts $\La_t$ follows because $Q_\kappa$ is boost invariant \cite{GrosseLechner:2007}.

As a consequence of translation covariance of $\ot_\kappa$ and translation invariance of the vacuum state, we have in particular 
\begin{align}\label{eq:translationinvariance}
 \om(F\ot_\kappa G) = \om(F\ot G)\,.
\end{align}

We also recall that the vacuum state is given by $n$-point functions the structure of which is fixed by Wick's Theorem. Namely, we have 
\begin{align}\label{vac}
    \langle\Om,\Psi(F)\rangle
    =
    \langle\Psi(F^*),\Om\rangle
    &=
    W_n(F),\qquad F\in\Ss(\Rl^{2n}),
\end{align}
where $F^*(x_1,\ldots,x_n)=\overline{F(x_n,\ldots,x_1)}$ and the $n$-point function $W_n\in\Ss'(\Rl^{2n})$ is best described by its Fourier transform
\begin{align}\label{eq:quasifree}
    \tilde{W}_n(p_1,\ldots,p_n)
    &=
    \begin{cases}
        0 & n \text{ odd}\\
        \sum\limits_{(\la,\mu)}\prod\limits_{k=1}^{n/2} \tilde{W}_2(p_{\la_k},p_{\mu_k}) & n \text{ even}
    \end{cases}
    .
\end{align}
Here the two-point function is given by, $p=(p^0,p^1),q=(q^0,q^1)\in\Rl^2$
\begin{align}
    \tilde{W}_2(p,q) = \frac{1}{\eps(p^1)}\,\delta(p^0-\eps(p^1))\,\delta(p+q),\qquad \eps(p^1)=\sqrt{(p^1)^2+1}
\end{align}
The sum $\sum_{(\la,\mu)}$ in \eqref{eq:quasifree} runs over all partitions $(\la,\mu)$ of $\{1,\ldots,n\}$ into $\frac{n}{2}$ disjoint pairs $(\la_k,\mu_k)$, $k=1,\ldots,n/2$, with $\la_k<\mu_k$. We will refer to the partitions $(\la,\mu)=\{(\la_1,\mu_1),\ldots,(\la_{n/2},\mu_{n/2})\}$ and their parts $(\la_k,\mu_k)$ as {\em contractions}.

\begin{theorem}\label{thm:limit-XY}
    Let $\kappa\neq0$. Let $X$ be a polynomial in the field operators $\phi_\kappa(f)$, $f\in\Ss(\Rl^2)$ and $Y'$ a polynomial in the field operators $\phi_{-\kappa}(g)$, $g\in\Ss(\Rl^2)$. Then, for any vectors $\Psi,\Psi'$ of finite particle number, we have 
    \begin{align}\label{eq:limit-formula}
        \lim_{t\to\pm\infty}\langle\Psi',\sigma_t(XY')\Psi\rangle 
        &= 
        \langle\Psi',\left(\om(XY')P_\Om+\om(X)\om(Y')\POm^\perp\right)\Psi\rangle.
    \end{align}
\end{theorem}
\begin{proof}
    As the left and right hand sides of \eqref{eq:limit-formula} are linear in $X$ and $Y'$, it is sufficient to consider field monomials, namely 
    \begin{align}
      X=\phi_\kappa(f_1)\cdots\phi_\kappa(f_n),\qquad Y'=\phi_{-\kappa}(g_1)\cdots\phi_{-\kappa}(g_m),
    \end{align}
    where $n,m\in\Nl_0$, and the functions $f_1,\ldots,f_n,g_1,\ldots,g_m\in C_c^\infty(\Rl)$ are arbitrary. 
    
    Furthermore, it is sufficient to consider vectors $\Psi,\Psi'$ that are not only of finite particle number but also in the Wightman domain $\D_W$ of the undeformed field $\phi_0$, namely 
    \begin{align}\label{eq:psipsi'}
        \Psi' = \Psi(l^*),\qquad \Psi=\Psi(r),
    \end{align}
    where $l\in\Ss(\Rl^{2a})$, $r\in\Ss(\Rl^{2b})$, and $a,b\in\Nl_0$ are arbitrary.
    
    Since the intersection $\D_W\cap\Hil_k$ with any $k$-particle space is dense in $\Hil_k$, and $\sigma_t(XY')$ is uniformly bounded on $\Hil_k$, it follows that \eqref{eq:limit-formula} holds for arbitrary finite particle number vectors $\Psi,\Psi'$ if we can establish it for $\Psi,\Psi'$ of the form \eqref{eq:psipsi'}.

    In the interest of a compact formula for the scalar product $\langle\Psi',\sigma_t(XY')\Psi\rangle$ in question \eqref{eq:limit-formula}, we introduce shorthand notations for the following Schwartz functions:
    \begin{align*}
        f^{\kappa,t}
        &:=
        \La_t^*f_1 \ot_\kappa \ldots \ot_\kappa \La_t^* f_n,\qquad 
        g^{-\kappa,t}
        :=
        \La_t^*g_1\ot_{-\kappa} \ldots \ot_{-\kappa}\La_t^* g_m.
     \end{align*}
    With these notations, we may repeatedly apply \eqref{eq:phikappa}, the associativity of $\ot_\kappa$, and \eqref{eq:translationinvariance}, to get
    \begin{align}\label{eq:WM1}
        \langle\Psi',\sigma_t(XY')\Psi\rangle
        &=
        \langle\Psi(l^*),\phi_\kappa(\La_t^* f_1)\cdots\phi_\kappa(\La_t^* f_n )\phi_{-\kappa}(\La_t^* g_1)\cdots\phi_{-\kappa}(\La^*_tg_m)\Psi(r)\rangle
        \\
        &=
        W_{n+m+a+b}(l\ot(f^{\kappa,t}\ot_\kappa(g^{-\kappa,t}\ot_{-\kappa}r)))\nonumber
        \\
        &=
        W_{n+m+a+b}((l\ot_\kappa f^{\kappa,t})\ot(g^{-\kappa,t}\ot_{-\kappa}r))) \nonumber
        .
    \end{align}
    For odd $N:=n+m+a+b$, \eqref{eq:WM1} vanishes. By distinguishing a few even/odd cases it is easy to see that in this case, the right hand side of \eqref{eq:limit-formula} also vanishes. For example, $\langle\Psi',\Om \rangle\langle\Om,\Psi\rangle\om(XY')$ can only be non-zero if $n+m$, $a$, and $b$ are all even, which is incompatible with $N$ being odd; and similar for the other terms in the right hand side of \eqref{eq:limit-formula}.
    
    For even $N$, \eqref{eq:WM1} equals a sum over contractions, namely
    \begin{align}
        \langle\Psi,\sigma_t(XY')\Psi'\rangle
        &=
        \sum_{(\la,\mu)} W_{(\la,\mu)}(t),
        \nonumber
        \\
        W_{(\la,\mu)}(t)
        &:=
        \int_{(\Rl^2)^N} dp 
        \,
        \tilde{l}(p_1,\ldots,p_a)\widetilde{f^{\kappa}}(\La_t p_{a+1},\ldots,\La_t p_{a+n})
        \nonumber
        \\
        &\qquad\qquad\quad\times \widetilde{g^{-\kappa}}(\La_t p_{a+n+1},\ldots,\La_t p_{a+n+m})\tilde{r}(p_{a+n+m+1},\ldots,p_N)
        \label{C}
        \\
        &\qquad\qquad\quad
        \times e^{ip(l)\cdot Q_\kappa p(f)}e^{-ip(g)\cdot Q_\kappa p(r)}\prod_{k=1}^{N/2} \tilde{W}_2(-p_{\la_k},-p_{\mu_k})
        \nonumber
        .
    \end{align}
    The exponential terms originate from the deformed tensor products \eqref{eq:tensor}, and as a shorthand notation, we have introduced the sums of the momenta $p_k\in\Rl^2$ appearing in the four functions, 
    \begin{align}\label{eq:momentumsums}
        p(l)=\sum_{k=1}^ap_k,\qquad p(f)=\sum_{k=a+1}^{a+n}p_k,\qquad p(g)=\sum_{k=a+n+1}^{a+n+m}p_k,\qquad p(r)=\sum_{k=a+n+m+1}^Np_k.
    \end{align}    
    For any contraction $(\la,\mu)$, the above integrand goes to zero pointwise as $t\to\pm\infty$. The $t$-dependence of the integrals $W_{(\la,\mu)}(t)$ depends however on the structure of $(\la,\mu)$, so that we have to distinguish a few different cases.
    
    To this end, we introduce the index sets
    \begin{align*}
        \I(l) &:= \{1,\ldots,a\},\\
        \I(f) &:= \{a+1,\ldots,a+n\},\\
        \I(g) &:= \{a+n+1,\ldots,a+n+m\},\\
        \I(r) &:= \{a+n+m+1,\ldots,a+n+m+b\},
    \end{align*}
    corresponding to variables of $\tilde l$, $\tilde f^\kappa$, $\tilde g^{-\kappa}$, and $\tilde r$, respectively.
    \begin{enumerate}
        \item[(I)] A contraction $(\la,\mu)$ is of type (I) if there exists $k\in\{1,\ldots,N/2\}$ such that either $\la_k$ or $\mu_k$ (but not both) lie in $\I(l)\cup\I(r)$. In this case, the contraction connects a variable of $\tilde f^\kappa$ or $\tilde g^{-\kappa}$ to a variable of $\tilde l$ or $\tilde r$.
        \item[(II)] A contraction $(\la,\mu)$ is of type (II) if it is not of type (I) and if for all $k\in\{1,\ldots,N/2\}$, the statement $\la_k\in\I(f)$ is equivalent to $\mu_k\in\I(f)$. Then also $\la_k\in\I(g)$ is equivalent to $\mu_k\in\I(g)$. In this case, we call the contraction $(\la_k,\mu_k)$ $f${\em-internal} and $g${\em-internal}, respectively.
        \item[(III)] A contraction $(\la,\mu)$ is of type (III) if it is not of type (I) or type (II), and if there exists $k\in\{1,\ldots,N/2\}$ such that $\la_k\in\I(l)$ and $\mu_k\in\I(r)$. In this case, $(\la,\mu)$ connects $\tilde f^\kappa$ and $\tilde g^{-\kappa}$, and also $\tilde l$ and $\tilde r$, but there are no contractions between $\tilde l,\tilde r$ and $\tilde f^\kappa,\tilde g^{-\kappa}$ as in type (I).
        \item[(IV)] A contraction $(\la,\mu)$ is of type (IV) if it is not of type (I), (II), or (III). In this case $(\la,\mu)$ connects $\tilde f^\kappa$ and $\tilde g^{-\kappa}$, but all variables of $\tilde l$ are contracted amongst themselves, and likewise for $\tilde r$.
    \end{enumerate}
    Clearly these cases are mutually exclusive and exhaust all possibilities.     

    \medskip
    
    \noindent {\bf (I)} Beginning our analysis of the four cases, we claim that 
    \begin{align}
        \lim_{t\to\pm\infty}W_{(\la,\mu)}(t)=0,\qquad (\la,\mu) \text{ of type (I)}.
    \end{align}
    In this case, the exponential factors $e^{ip(l)\cdot Q_\kappa p(f)}e^{-ip(g)\cdot Q_\kappa p(r)}$ in \eqref{C} are irrelevant and we may estimate \eqref{C} by triangle inequality. After carrying out all integrations over the delta distributions in the two-point functions $\tilde W_2(p,q)=\eps(p^1)^{-1}\delta(p^0-\eps(p^1))\delta(p+q)$, setting the momenta $p_{\la_k}$, $k=1,\ldots,N/2$, equal to $-p_{\mu_k}$ and restricting the $p_{\mu_k}$ to the upper mass shell, i.e. constraining them to the form $p_{\mu_k}=(\eps(p_{\mu_k}^1),p_{\mu_k}^1)$, we may then substitute $\sinh\te_k:=p_{\mu_k}^1$, $dp_{\mu_k}^1/d\te_k=\eps(p_{\mu_k}^1)$. Since the boosts $\La_t$ act by shifts in the rapidities $\te_k$, this shows that we have
    \begin{align}
        |W_{(\la,\mu)}(t)|
        &\leq 
        \int\limits_{\Rl^N} F(\bte,\bte'-t,\bte''-t,\bte''')\,\prod_{k=1}^{N/2}\delta(\te_{\la_k}-\te_{\mu_k})d\te_{\la_k}d\te_{\mu_k}, 
    \end{align}
    where $F\in\Ss(\Rl^N)$ is a Schwartz function, $\bte$, $\bte'$, $\bte''$, and $\bte'''$ are the rapidity transforms of the variables of $\tilde l$, $\tilde f^\kappa$, $\tilde g^{-\kappa}$, and $\tilde r$, respectively, and the shorthand notations $\bte'-t  := (\te_1'-t,\ldots,\te_n'-t)$ and $\bte''-t  := (\te_1''-t,\ldots,\te_m''-t)$ have been introduced.
    
    We now estimate $F$ in terms of Schwartz seminorms in order to obtain 
    \begin{align*}
        |W_{(\la,\mu)}(t)|
        &
        \leq
        C\int\limits_{\Rl^{N/2}} \prod_{\alpha}(1+\te_\alpha^2)^{-1}(1+(\te_\alpha-t)^2)^{-1}\cdot \prod_{\beta}(1+\te_\beta^2)^{-2} \;d\bte.
    \end{align*}
    Here the first product arises from those contractions $(\la_k,\mu_k)$ that link $t$-dependent and $t$-independent variables; by the type (I) assumption this product is not empty. The second product arises from those contractions $(\la_k,\mu_k)$ that link either two $t$-independent or two $t$-dependent variables. In the above form, it is then clear that dominated convergence can be applied, and we arrive at the claimed limit $\lim\limits_{t\to\pm\infty}W_{(\la,\mu)}(t)=0$.
    
    \medskip
    
    \noindent{\bf (II)} We move on to type (II) contractions, and claim that in this case $W_{(\la,\mu)}(t)$ does not depend on $t$. 
    
    By definition of type (II), all variables of $\tilde f^\kappa$ are contracted amongst themselves, and likewise for $\tilde g^{-\kappa}$. In view of the energy momentum delta distributions $\delta(p+q)$ in the two-point function, in this case we observe that on the support of $\prod_{k=1}^{N/2}\tilde{W}_2(p_{\la_k},p_{\mu_k})$, the sums $p(f)$ and $p(g)$ \eqref{eq:momentumsums} vanish. Thus the exponential terms in \eqref{C} drop out and the integrand becomes a factor of three functions in independent variables: a) $\tilde l(p_1,\ldots,p_a)\tilde r(p_{a+n+m+1},\ldots,p_N)$, b) $\tilde f^\kappa(\La_t p_{a+1},\ldots,\La_t p_{a+n})$, and c) $\tilde g^{-\kappa}(\La_t p_{a+n+1},\ldots,\La_t p_{a+n+m})$. In view of the structure of a type (II) contraction, also $\prod_{k=1}^{N/2}\tilde{W}_2(p_{\la_k},p_{\mu_k})$ splits in the same manner so that $W_{(\la,\mu)}(t)$ factors into a product of three integrals. Now the substitutions $p_j\mapsto\La_{-t}p_j$ in the integrals over $\tilde f^\kappa$ and $\tilde g^{-\kappa}$ eliminate all $t$-dependence, i.e. $W_{(\la,\mu)}(t)$ is given by
    \begin{align*}
        &
        \int \tilde{l}(p_1,\ldots,p_a)\tilde r(p_{a+n+m+1},\ldots,p_N)\prod_{k\atop \{\la_k,\mu_k\}\subset\I(l)\cup\I(r)} \tilde W_2(-p_{\la_k},-p_{\mu_k})dp_{\la_k}dp_{\mu_k}
        \\
        &\quad\times
        \int \tilde{f^\kappa}(p_{a+1},\ldots,p_{a+n})\prod_{k\atop\{\la_k,\mu_k\}\subset\I(f)} \tilde W_2(-p_{\la_k},-p_{\mu_k})dp_{\la_k}dp_{\mu_k}
        \\
        &\quad\times
        \int \tilde{g^{-\kappa}}(p_{a+1},\ldots,p_{a+n})\prod_{k\atop\{\la_k,\mu_k\}\subset\I(g)} \tilde W_2(-p_{\la_k},-p_{\mu_k})dp_{\la_k}dp_{\mu_k}
    \end{align*}
    By this factorization, the sum over all type (II) contractions decomposes into a product of three sums, running over all contractions of $\I(l)\cup\I(r)$, $\I(f)$, and $\I(g)$, respectively. 
    
    By Wick's Theorem, the first sum (contractions of $\I(l)\cup\I(r)$) coincides with $\om(\phi_0(l)\phi_0(r))=\langle\Psi(l^*),\Psi(r)\rangle=\langle\Psi',\Psi\rangle$. The contractions of $\I(f)$ sum to $W_n(f^\kappa)=\om(\phi_\kappa(f_1)\cdots\phi_\kappa(f_n))=\om(X)$, and analogously the contractions of $\I(g)$ sum to $\om(Y')$. 
    
    Hence we obtain 
    \begin{align}
        \sum_{(\la,\mu)\,\text{type (II)}} W_{(\la,\mu)}(t)
        &=
        \om(X)\om(Y')\langle\Psi',\Psi\rangle.\nonumber
    \end{align}
    
    \noindent{\bf (III)} We claim 
     \begin{align}
        \lim_{t\to\pm\infty}W_{(\la,\mu)}(t)=0,\qquad (\la,\mu) \text{ of type (III)}.
    \end{align}
    The initial part of the following arguments works for both, type (III) and type (IV). As there are no contractions connecting $\tilde f^\kappa,\tilde g^{-\kappa}$ with $\tilde l,\tilde r$ for these types, we can again remove the $t$-dependence in the variables of $\tilde f^\kappa$ and $\tilde g^{-\kappa}$ by substituting $p_j\mapsto\La_{-t}p_j$. This substitution does not introduce $t$-dependence in the variables of $\tilde l$ and $\tilde r$, but only in the exponentials, resulting in the factor $\exp(ip(l)\cdot Q_\kappa\La_{-t}p(f))\exp(-i\La_{-t}p(g)\cdot Q_\kappa p(r))$. Since all variables $p_k$ with $k\in\I(f)\cup\I(g)$ are contracted amongst themselves, we have $p(f)+p(g)=0$ and $p(l)+p(r)=0$ on the support of the delta distributions. Taking into account the antisymmetry of $Q_\kappa$, this simplifies the exponential factor to $\exp(2ip(l)\cdot Q_\kappa\La_{-t}p(f))$.
    
    After these remarks, we switch again to rapidity parameterization as for type~(I) and obtain 
    \begin{align}\label{III+IV}
     W_{(\la,\mu)}(t)
     =
     \int d\te d\te' \,L(\te_1,\ldots,\te_{(a+b)/2})F(\te_1',\ldots,\te_{(n+m)/2}')\prod_{j,k}e^{2i\kappa\sinh(\te_j-\te_k'+t)}.
    \end{align}
    Here $L\in\Ss(\Rl^{(a+b)/2})$ arises from $\tilde l\ot \tilde r$ by rapidity parameterization and identifying variables via integration of delta distributions (hence only half of the original $a+b$ many variables remain). The function $F\in\Ss(\Rl^{(n+m)/2})$ arises from $\tilde f^{\kappa}\ot \tilde g^{-\kappa}$ in the same way. 
    
    The exponential factor $\prod_{j,k}e^{2i\kappa\sinh(\te_j-\te_k'+t)}$ in \eqref{III+IV} is obtained from the previously described $\exp(2ip(l)Q\La_{-t}p(f))$ as follows: First we note that contractions $(\la_k,\mu_k)$ that are $f$-internal (i.e. $\{\la_k,\mu_k\}\subset\I(f)$) have the effect that the partial summand $p_{\la_k}+p_{\mu_k}=0$ vanishes from $p(f)$ \eqref{eq:momentumsums}, and analogously for $l$-internal contractions and $p(l)$. Hence the exponential term takes the form $\exp(2ip'(l)Q_\kappa\La_{-t}p'(f))$, where the primes indicate that the sums \eqref{eq:momentumsums} run only over those momenta that are not removed by $l$-internal or $f$-internal contractions. 
    
    As $p(\te)\cdot Q_\kappa\La_{-t}p(\te')=\kappa\sinh(\te-\te'+t)$, this explains the product $\prod_{j,k}$ in \eqref{III+IV} which runs over the remaining not self-contracted variables.
    
    Since we are not in type (II), at least one such variable of $\tilde f^\kappa$ is contracted with a variable of $\tilde g^{-\kappa}$, i.e. the product over $k$ is not empty. 
    
    Now we distinguish between type (III) and (IV). In type (III), the product over $j$ is also not empty (that is, $(\la,\mu)$ connects $\tilde l$ and $\tilde r$). We claim that 
    \begin{align}\label{IIIlimit}
        \lim_{t\to\pm\infty}W_{(\la,\mu)}(t)=0,\qquad (\la,\mu) \text{ of type (III)}.
    \end{align}
    To show this, we use a Riemann-Lebesgue type argument and use integration by parts to exploit the oscillatory term $\prod_{j,k}e^{2i\kappa\sinh(\te_j-\te'_k+t)}$; this is the point were $\kappa\neq0$ enters the proof of the theorem. We pick some $j$ occuring in this product, and rewrite the preceding integral as
    \begin{align*}
     |W_{(\la,\mu)}(t)|
     &=
     \left|\int d\te d\te' \,\frac{L(\te_1,\ldots,\te_{(a+b)/2})F(\te_1',\ldots,\te_{(n+m)/2}')}{\sum_{k}2i\kappa\cosh(\te_j-\te_k+t)}\frac{\partial}{\partial\te_j}\prod_{j,k}e^{2i\kappa\sinh(\te_j-\te_k'+t)}
     \right|
     \\
     &\leq
     \int d\beta d\te \,\left|F(\te'_1,\ldots,\te'_{(n+m)/2})
     \frac{\partial}{\partial\te_j}\frac{L(\te_1,\ldots,\te_{(a+b)/2})}{\sum_{k}2i\kappa\cosh(\te_j-\te'_k+t)}
     \right|
    \end{align*}
    Expanding the derivatives and using $\left|\frac{\partial}{\partial\te_j}\frac{1}{\cosh(\te_j-\te_k'+t)}\right|\leq\frac{1}{\cosh(\te_j-\te_k+t)}$, it is then easy to see that the integrand goes to zero pointwise as $t\to\pm\infty$. Since $\cosh x\geq1$, we have an integrable majorant and may apply dominated convergence to conclude the claimed limit \eqref{IIIlimit}.
    
    \medskip 
    
    \noindent{\bf (IV)} Type (IV) contributions are also described by \eqref{III+IV}. But in contrast to type (III), there is no contraction between $\tilde l$ and $\tilde r$ in this case, which results in the product $\prod_{j,k}$ in \eqref{III+IV} dropping out. So as in type (II), it follows that $W_{(\la,\mu)}(t)$ is independent of $t$ in type (IV), and the sum over all type (IV) contractions splits into three sums: a) the contractions of all variables of $\tilde l$, b) the contractions of all variables of $\tilde r$, and c) the contractions of all variables of $\tilde f^\kappa$ and $\tilde g^{-\kappa}$. 
    
    The sum a) yields $\om(\phi_0(l_1)\cdots\phi_0(l_a))=\langle\Psi(l^*),\Om\rangle=\langle\Psi',\Om\rangle$, and similarly the sum b) results in $\langle\Psi(r),\Om\rangle=\langle\Psi,\Om\rangle$. 
    
    To evaluate the sum c), note that we are here summing only those contractions between $\tilde f^\kappa$ and $\tilde g^{-\kappa}$ that are connecting $\tilde f^\kappa$ and $\tilde g^{-\kappa}$ and not those contractions that are $f$-internal and $g$-internal. Hence the sum c) yields 
    \begin{align*}
        \om(\phi_\kappa&(f_1)\cdots\phi_\kappa(f_n)\phi_{-\kappa}(g_1)\cdots\phi_{-\kappa}(g_m))
        \\
        &-
        \om(\phi_\kappa(f_1)\cdots\phi_\kappa(f_n))\cdot\om(\phi_{-\kappa}(g_1)\cdots\phi_{-\kappa}(g_m))
        =
        \om(XY')-\om(X)\om(Y'),
    \end{align*}
    and we obtain 
   \begin{align*}
        \sum_{(\la,\mu)\,\text{type (IV)}} W_{(\la,\mu)}(t) = (\om(XY')-\om(X)\om(Y'))\langle\Psi,\Om\rangle\langle\Om,\Psi'\rangle.
    \end{align*}
    Summing all types (I)--(IV), we arrive at 
    \begin{align*}
        \lim_{t\to\pm\infty}\langle\Psi',\sigma_t(XY')\Psi\rangle 
        &= 
        \langle\Psi',\left((\om(XY')-\om(X)\om(Y'))P_\Om+\om(X)\om(Y')1\right)\Psi\rangle
        \\
        &= 
        \langle\Psi',\left(\om(XY')P_\Om+\om(X)\om(Y')(1-\POm)\right)\Psi\rangle.
    \end{align*}
    This finishes the proof of the theorem.
\end{proof}

It should be noted that the proof of Theorem~\ref{thm:limit-XY} did not use any specific support properties of $f$ and $g$, i.e. strict localization of the operators $X,Y'$ was not required. Also we did not only evaluate the limit $t\to-\infty$, corresponding to scaling points in $x\in\Rl\backslash\{0\}$ to $\infty$ (as in Prop.~\ref{prop:triviality-with-limits}), but also the opposite limit $t\to+\infty$, which corresponds to scaling $x$ to $0$.

\medskip 

We next transfer the limiting behaviour established in Theorem~\ref{thm:limit-XY} to arbitrary (bounded) operators in the von Neumann algebras $\M_\kappa$, ${\M_\kappa}'$ \eqref{Mkappa}.

\begin{theorem}\label{thm:boundedlimit}
    Consider the Borchers triple $(\M_\kappa,T,\Om)$ defined in \eqref{Mkappa} and let $\kappa>0$. Then, for any $A\in\M_\kappa$, $B\in{\M_\kappa}'$,
    \begin{align}\label{eq:weak-limit-formula}
        \wlim_{t\to\pm\infty} \sigma_t(AB) &= 
            \om(AB)\POm+\om(A)\om(B)\POm^\perp
        .
    \end{align}
\end{theorem}

\noindent{\em Remark:} Note that by setting $A=1$ or $B=1$, \eqref{eq:weak-limit-formula} reproduces $\sigma_t(A)\to\om(A)1$ and $\sigma_t(B)\to\om(B)1$ in accordance with the fact that these limits are contained in the type III${}_1$ factors $\M$ and $\M'$ and fixed points of the modular group, thus trivial.

\begin{proof}
    As shorthand notations, we will write $\Psi_A=A\Om$ and $P_{AB}=\om(AB)\POm+\om(A)\om(B)\POm^\perp$ for the expression appearing in \eqref{eq:weak-limit-formula}, and similarly for various other operators. Note that $P_{AB}=\sigma_t(P_{AB})$ is invariant under the modular group. 
    
    The idea of the proof is to consider expectation values $\langle\Psi',\sigma_t(AB-P_{AB})\Psi\rangle$ in special vectors, namely vectors of the form $\Psi'=\Psi_{L^*}$, $\Psi=\Psi_R$, where $L,R$ are closed operators affiliated with ${\M_\kappa}'$ (``left'') and $\M_\kappa$ (``right''), respectively, and with domains such that $\Om\in\D(L^*)\cap\D(R)$. As $\M_\kappa$ and $\M'_\kappa$ are stable under the modular group, this implies $\sigma_t(A^*)\Om\in\D(L^*)$, $\sigma_t(B)\in\D(R)$, and $L^*\sigma_t(A^*)\Om=\sigma_t(A^*)L^*\Om$, $R\sigma_t(B)\Om=\sigma_t(B)R\Om$, so that we may rewrite the expectation value in question as 
    \begin{align*}
        \langle \sigma_t(A^*)\Psi_{L^*},\sigma_t(B)\Psi_R\rangle
        &=
        \langle L^*\sigma_t(A^*)\Om,R\sigma_t(B)\Om\rangle
        =
        \langle \sigma_{-t}(L^*)\Psi_{A^*},\sigma_{-t}(R)\Psi_{B}\rangle.
    \end{align*}
    This symmetry in the roles of $A,B$ and $L,R$ also holds for the claimed limit $\langle\Psi',P_{AB}\Psi\rangle$, namely
     \begin{align*}
        \langle \Psi_{L^*},P_{AB}\Psi_R\Om\rangle
        &=
        \om(AB)\om(L)\om(R)+\om(A)\om(B)\om(LR)-\om(A)\om(B)\om(L)\om(R)
        \\
        &=
        \langle\Psi_{A^*},P_{LR}\Psi_{B}\rangle.
    \end{align*}
    Hence
    \begin{align}\label{eq:symmetrylimit}
        \langle\sigma_t(A^*)\Psi_{L^*},\sigma_t(B)\Psi_R\rangle-\langle\Psi_{L^*},P_{AB}\Psi_R\Om\rangle
        =
        \langle\sigma_{-t}(L^*)\Psi_{A^*},\sigma_{-t}(R)\Psi_{B}\rangle-\langle \Psi_{A^*},P_{LR}\Psi_{B}\rangle.
    \end{align}
    We mention as an aside that this symmetric formula implies that in case the limit \eqref{eq:weak-limit-formula} holds for $t\to-\infty$, then it also holds for $t\to+\infty$, and vice versa.
    
    In a first step, we will choose $L,R$ to be field operator polynomials as in Prop.~\ref{thm:limit-XY}, but smeared with test functions supported on the left and right, so that $L$ is affiliated with ${\M_\kappa}'$ and $R$ is affiliated with ${\M_\kappa}$. The operators $A,B$ are taken to be bounded and (for technical reasons) smooth, i.e. $A\in{\M_\kappa}^\infty$, $B\in{{\M_\kappa}'}^\infty$. We mention in passing that due to the localization properties of $\M_\kappa,{\M_\kappa}'$, the smooth elements of these von Neumann algebras form strongly dense ${}^*$-subalgebras \cite{BuchholzLechnerSummers:2011}. Generalizing the limit formula \eqref{eq:limit-formula} of Prop.~\ref{thm:limit-XY} to vectors $\Psi',\Psi$ that are not necessarily of finite particle number, but rather of the form $\Psi'=\Psi_{A^*}$, $\Psi=\Psi_B$, we will show below that the right hand side of \eqref{eq:symmetrylimit} converges to $0$ as $t\to\pm\infty$.
    
    Postponing the proof of this part for a moment, let us explain how it implies the conclusion of the theorem. Taking into account \eqref{eq:symmetrylimit} and the fact that the vectors $\Psi_{L^*}$, $\Psi_R$ range over dense subspaces of $\Hil$ as $L,R$ vary within the limitations explained above (Reeh-Schlieder property), and $\sigma_{t}(AB-P_{AB})$, $t\in\Rl$, is uniformly bounded in operator norm, the claimed weak limit $\sigma_t(AB)\to P_{AB}$ follows immediately for smooth $A,B$.
    
    To eliminate the assumption of smoothness, we then consider \eqref{eq:symmetrylimit} once more, this time with smooth $A,B$ as before, and bounded, not necessarily smooth $L\in{\M_\kappa}'$, $R\in\M_\kappa$. At this stage we know that the left hand side of \eqref{eq:symmetrylimit} goes to zero as $t\to\pm\infty$, and hence the desired limit $\langle\Psi_{A^*},\sigma_{-t}(LR-P_{LR})\Psi_B\rangle\to0$ holds (right hand side of \eqref{eq:symmetrylimit}). But as the smooth subalgebras ${\M_\kappa}^\infty\subset\M_\kappa$ and ${{\M_\kappa}'}^\infty\subset{\M_\kappa}'$ are strongly dense, they have $\Om$ as a cyclic vector. Hence the limit carries over to arbitrary vectors on the left and right hand sides of the scalar product, and the proof is finished. 
    
    \medskip 
    
    It remains to show that the right hand side of \eqref{eq:symmetrylimit} converges to $0$ for field operator polynomials $L,R$ as in Prop.~\ref{thm:limit-XY} and arbitrary smooth $A\in{\M_\kappa}^\infty$, $B\in{\M_\kappa'}^\infty$. To that end, let $Q_n:=P_0\oplus\ldots\oplus P_n$ denote the orthogonal projection onto the subspace of particle number at most $n$ in our Fock space~$\Hil$. The projections $Q_n^\perp$ leave the domain of $L^*$ invariant, and we may estimate according to 
    \begin{align*}
        |\langle \sigma_t(L^*)\Psi_{A^*},&\sigma_t(R)\Psi_B\rangle-\langle\Psi_{A^*},P_{LR}\Psi_B\rangle|
        \\
        &\leq
        |\langle Q_n\Psi_{A^*},\sigma_t(LR-P_{LR})\Psi_A\rangle|
        +
        |\langle Q_n^\perp\Psi_{A^*},(\sigma_t(LR)-P_{LR})\Psi_B\rangle|
    \end{align*}
    Since $L,R$ are field operator polynomials, they change particle number only by a finite amount, i.e. there exists $m\in\Nl$ such that for all $n\in\Nl,t\in\Rl$
    \begin{align*}
        |\langle &Q_n\Psi_{A^*},\sigma_t(LR-P_{LR})\Psi_B\rangle|
        +
        |\langle Q_n^\perp\Psi_{A^*},\sigma_t(LR-P_{LR})\Psi_B\rangle|
        \\
        &=
        |\langle Q_n\Psi_{A^*},\sigma_t(LR-P_{LR})Q_{n+m}\Psi_B\rangle|
        +
        |\langle Q_n^\perp\Psi_{A^*},\sigma_t(LR-P_{LR})Q_{n-m}^\perp \Psi_B\rangle|
        \\
        &\leq
        |\langle Q_n\Psi_{A^*},\sigma_t(LR-P_{LR})Q_{n+m}\Psi_B\rangle|
        \\
        &\qquad 
        +
        \|\sigma_t(L^*)Q_n^\perp\Psi_{A^*}\|\|\sigma_t(R)Q_{n-m}^\perp \Psi_B\|+\|Q_n^\perp\Psi_{A^*}\|\|P_{LR}\|\|Q_{n-m}^\perp \Psi_B\|.
    \end{align*}
    The first term contains only vectors of finite particle number on the left and right and therefore converges to $0$ as $t\to\pm\infty$ by Prop.~\ref{thm:limit-XY}. The third term goes to zero as $n\to\infty$ because $Q_n^\perp,Q_{n-m}^\perp\to0$ strongly in this limit. Hence it is sufficient to show that the two norms $ \|\sigma_t(L^*)Q_n^\perp\Psi_{A^*}\|$, $\|\sigma_t(R)Q_{n-m}^\perp \Psi_B\|$ go to zero as $n\to\infty$, uniformly in $t\in\Rl$.
    
    At this point, our smoothness assumption on $A$ and $B$ enters. Namely, since $A,B$ are smooth, the vectors $\Psi_{A^*},\Psi_B$ lie in particular in $\bigcap_{k\geq0}\D(\Pst_0^k)$, where $\Pst_0$ is the generator of the time translations, i.e. the second quantization of $\frac{1}{2}(\Pst+\Pst^{-1})\geq1$. This implies $\Pst_0\geq N$ (the particle number operator, the second quantization of the identity), and hence $\Psi_{A^*},\Psi_B\in\bigcap_{k\geq0}\D(N^k)$. On the other hand, there exists $k\in\Nl$ such that $L^* N^{-k}$ and $RN^{-k}$ are bounded because each single field operator $\phi_{\pm\kappa}(f)$, $f\in\Ss(\Rl^2)$, satisfies $\|\phi_{\pm \kappa}(f)N^{-1/2}\|<\infty$. Choosing $k$ large enough in this manner, and observing that $Q_n^\perp$ and $\Delta^{it}$ commute with $N$, we can estimate according to
    \begin{align*}
        \|\sigma_t(L^*)Q_n^\perp\Psi_{A^*}\|
        =
        \|L^* N^{-k}\Delta^{-it}Q_n^\perp N^k\Psi_{A^*}\|
        \leq
        \|L^*N^{-k}\|\|Q_n^\perp N^k\Psi_{A^*}\|,
    \end{align*}
    which goes to zero as $n\to\infty$, uniformly in $t$, as required. The second norm $\|\sigma_t(R)Q_{n-m}^\perp \Psi_B\|$ can be estimated analogously.
\end{proof}

To conclude singularity of the relative commutant $\N_\kappa'\cap\M_\kappa$ from this we need the following simple lemma.

\begin{lemma}\label{lemma:noproductstate}
    Let $(\N\subset\M,\Om)$ be a half-sided modular inclusion on a Hilbert space~$\Hil$ of dimension $\dim\Hil>1$. Then $\om=\langle\Om,\,\cdot\,\Om\rangle$ is not a product state on $\N\vee J\N J$.
\end{lemma}
\begin{proof}
    Let $A\in\N$, $B\in J\N J$. Then also $\alpha_x(A)\in\N$ for any $x>0$. If $\om$ was a product state, we would therefore have $\langle A^*\Om,T(-x)B\Om\rangle=\om(\alpha_x(A)B)=\om(A)\om(B)=\langle A^*\Om,\POm B\Om\rangle$ for $x>0$. As $\Om$ is cyclic for $\N$ and $J\N J$, this would imply $T(-x)=\POm$, which is only possible if $\POm=1$, i.e. $\Hil=\Cl\Om$ is one-dimensional.
\end{proof}

\noindent{\em Proof of Theorem~\ref{thm:main}}
As recalled earlier, the statement for $\kappa=0$ is a consequence of the second quantization structure of $\M_0=\M(H)$. For $\kappa>0$, we pick $A\in\N_\kappa$, $B\in J\N_\kappa J$. According to Lemma~\ref{lemma:limits} and Theorem~\ref{thm:boundedlimit}, 
\begin{align*}
 P_{AB}=(\om(AB)-\om(A)\om(B))\POm+\om(A)\om(B)1 \in \scrX_\kappa
\end{align*}
lies in the algebra at infinity. In view of Lemma~\ref{lemma:noproductstate}, we may choose $A,B$ in such a way that the coefficient $\om(AB)-\om(A)\om(B)$ in front of~$\POm$ is not zero. Then $$S=\frac{AB-\om(A)\om(B)1}{\om(AB)-\om(A)\om(B)}\in\N_\kappa\vee J\N_\kappa J$$ satisfies $\sigma_t(S)\to\POm$ weakly as $t\to-\infty$, and Prop.~\ref{prop:triviality-with-limits}~c) gives the result. \hfill$\square$

\section{Conclusion and Outlook}\label{section:conclusion}

After the free product construction of Longo, Tanimoto and Ueda \cite{LongoTanimotoUeda:2019}, the half-sided inclusions constructed in this paper provide further and arguably simpler examples of singular Borchers triples, giving further insight into the structure of the family of all half-sided modular inclusions. 

Although Prop.~\ref{prop:triviality-with-limits} provides a necessary and sufficient condition for a half-sided inclusion to be singular, it seems difficult to use it in order to prove non-singularity. So the ideas presented here consist first and foremost in a method for constructing ``counterexamples'' to the preferred standard case \Std. However, this method might well inform complementary methods that are designed to {\em exclude} the singular case \Sing, or help settle the status of the various undecided examples. 

In this work, it became apparent that half-sided inclusions are very sensitive to deformations and their relative commutant can vary ``discontinuously'' with a deformation parameter. It also seems likely that there exist many singular half-sided inclusions. In view of the weak limit formula \eqref{eq:weak-limit-formula} that is generic for any $A\in\M_\kappa$, $B\in{\M_\kappa}'$, it is not unreasonable to expect that any standard half-sided inclusion $\La=(\N\subset\M,\Om)$, not necessarily arising from a standard pair by second quantization, can be deformed to a one-parameter family of half-sided inclusions $(\La_\kappa)_{\kappa\geq0}$ such that\footnote{Work in progress.} $\La_0=\La$ and $\La_\kappa$ is singular for every $\kappa>0$.

\section*{Acknowledgements}

Support by the German Research Foundation DFG through the Heisenberg project ``Quantum Fields and Operator Algebras'' (LE 2222/3-1) is gratefully acknowledged.

\end{document}